\documentclass[12pt]{article}
\usepackage{amsthm,amssymb,amsmath,mathrsfs}
\usepackage{epsfig}
\usepackage{pstricks,pst-node,amssymb,amsmath,graphics,latexsym,tabularx,shapepar}
\usepackage[all,cmtip,2cell,line]{xypic}
\UseAllTwocells\SilentMatrices
\usepackage{graphicx}
\newtheorem{prelem}{{\bf Theorem}}

 \newtheorem{theorem}{Theorem}
\newtheorem{corollary}[theorem]{Corollary}
\newtheorem{lemma}[theorem]{Lemma}
\newtheorem{proposition}[theorem]{Proposition}
\theoremstyle{definition}
\newtheorem{definition}[theorem]{Definition}
\theoremstyle{remark}

\theoremstyle{conjecture}

\title{Generalized geometric Hamilton-Jacobi theorem on Lie algebroids}
\date{}
\author
{Gh. Haghighatdoost$^1$, R. Ayoubi$^2$ \\
Department of Mathematics\\ Azarbaijan Shahid Madani University\\
Tabriz, Iran\\{ \tt $^1$gorbanali@azaruniv.ac.ir}\\ {\tt $^2$rezvaneh.ayoubi@azaruniv.ac.ir}}

\begin{document}

\maketitle

\begin{abstract}
In this paper, some of formulations of Hamilton-Jacobi equations for Hamiltonian system on Lie algebroids are given. Here we use the general properties of Lie algebroids to express and prove two geometric version of the Hamilton-Jacobi theorem for Hamiltonian system on Lie algebroids. Then this results are generalized and two types of time-dependent Hamilton-Jacobi theorem of Hamiltonian system on Lie algebroids are obtained.\\

\noindent{\bf AMS}: 70H20, 70H06, 17B66.  \\

\noindent {\bf keywords}: Lie algebroids, Hamilton-Jacobi theory, Hamiltonian system.

\end{abstract}
\maketitle

\section{Introduction}
In 1820, Hamilton developed the principles of the Hamilton-Jacobi theory for issues in wave and geometric obtics. His intention was to put forward a complete theory for obtics as Lagrange dictated in the years before. His work on obtics was completed until the age of 22 and he present a comprehensive theory of particles and light-wave concepts. After that, Hamilton develop his work on mechanics. In his Second Essay on a General Method in Dynamics (1835), he has raised both Hamilton-Jacobi equation and Hamilton's canonical equations.
Although the integration of partial differential equations is more difficult than solving arbitrary equations, the Hamilton-Jacobi theory was expressed to be useful tool for studying issues of obtics, mechanics and geometry.\\
In many cases, the Hamilton-Jacobi equation provides a method for integrating the motion equations of the system, even when the Hamiltonian system does not respond or is not completely solvable. Therefore, the Hamilton-Jacobi theory has been very much considered and developed in recent years.\\
It is obvious that Hamilton-Jacobi theory from the variational point of view is originally developed by Jacobi in 1866, which state that the integral of Lagrangian of a system along the solution of its Euler-Lagrange equation satisfies the Hamilton-Jacobi equation [3].
Abraham and Marsden described this issue geometrically as a theorem, see theorem 5.2.4 in [1]. This theorem describes solutions to Hamilton-Jacobi equation by finding particular integral curves of Hamiltonian vector field, which those integral curves are obtained, as possible, from integral curves of a vector field on smooth manifold Q.\\
Hong Wang proved two types of geometric Hamilton-Jacobi theorem for a Hamiltonian system on the cotangent bundle for a smooth manifold, see [3]. In 2004, Manuel De Leon, Juan C.Marrero and Eduardo Martinez [2], proved the Abraham and Marsden's theorem for Lie algebroids.\\
Now, the problem is, how to generalize  the Hamilton-Jacobi theory for Hamiltonian system on Lie algebroids. In other words, our work is devoted to prove two types of Hamilton-Jacobi theorem for Hamiltonian system on Lie algebroids by using the Manuel De Leon, Juan C.Marrero and Eduardo Martinez's theorem and Hong Wang's work.\\
This paper is organized as follows. In next section we summarize some of the main concepts of Lie algebroids. In section 3 we present a very useful Lemma, which we use to prove our basic theorems, then we prove two types of geometric Hamilton-Jacobi theorem for Hamiltonian system on Lie algebroids and present two examples. In section 4 we generalize the above results to the non-autonomous case and obtain two types of Hamilton-Jacobi theorem for the time-dependent Hamiltonian system.

\section{Some Basic Concepts and Definitions}

\subsection{Lie algebroids}
\begin{definition}
\label{dfn-algebroid}
A Lie algebroid over a manifold $M$ is a vector bundle $ E $  of rank $ n $ over
$M$ equipped with a Lie
algebra structure $[~,~]$ on its space of sections and a bundle map
$\rho : E \rightarrow TM$ called the anchor,
which induces a Lie algebra homomorphism (also denoted $\rho $)
from sections of $E$ to vector fields on $M$.
The identity $$ [f X , Y] = f[X ,Y ]+\rho \left( X \right)(f)Y \enspace $$
must be satisfied for every smooth function $f$ on $M$.
\end{definition}

The standard local coordinates on $E $ have the form $(x,\lambda)$
where the $x_{i}$'s are coordinates on the base $M$ and the $ \lambda
_{i}$'s are linear coordinates on the fibres , associated with a basis
$ X_{i} $ of sections of the Lie algebroid.  In terms of such
coordinates, the bracket and anchor have expressions $ [X _{i},X
_{j}]=\sum c_{ij} ^{k} X _{k} $ and $ \rho (X_{i})= \sum a_{ij} \frac{\partial}{\partial x_{j}},$
 where the $c_{ij} ^{k}$ and $ a_{ij} $  are structure functions lying is $ C^{\infty}(M)$.
 
\subsection{Examples}
\label{subsec-firstexamples}
The tangent bundle $TQ$ on $ Q $ is the basic example of Lie algebroid over $ Q $,\\
 $ \tau_{Q} : TQ \longrightarrow Q $, with the identity mapping as anchor.  Also any integrable subbundle of $TQ$ is a Lie algebroid with the inclusion as anchor and the induced bracket.  Therefore any Lie algebra $\frak g $ is a Lie algebroid over a point.\\
More generally, if $Q$ is a principal $G$-bundle over $M=Q/G,  \pi : Q \longrightarrow M $, then $TP/G$ is a vector bundle over $M, \tau_{Q \vert G} : TQ/G \longrightarrow Q/G $,  whose sections are the $G$-equivariant vector fields (infinitesimal gauge transformations) on $Q$. These sections inherit the bracket from $TP$, and the derivative $TP \rightarrow TM$ of the projection from $Q$ to $M$  passes to a map from $ TP/G $ to $TM$ which is the anchor of a Lie algebroid structure on $ TP/G $, $ \rho ( [v_{q} ] ) = ( T_{q} \pi ) (v_{q}) $  which we call the {\bf gauge algebroid} of $Q$.\\
Another example comes from actions of Lie algebras on manifolds. If $ \varphi $ be an action of Lie algebra $ \frak g $ on $ M $, we denote the infinitesimal action of $ \frak g $ on manifold $ M $ by $ \xi_{M}:=d \varphi_{x} (\xi): \frak g \longrightarrow \chi (M) $ the {\bf action algebroid} is the
trivial bundle $M \times \frak g $ where $ \tau:M\times \frak g \longrightarrow M $ is projection over the first factor, with the anchor
 $\rho (x,\xi)=\xi_{M} (x)=d\varphi_{x}(\xi)$ and the bracket $[X
,Y ](x)=[X (x),Y (x)]+d\varphi_{x} (X )\cdot Y -d\varphi_{x} (Y )\cdot X.$ \\
In the bracket formula, we have identified sections of $M\times
\frak g $ with $\frak g $-valued functions on $M$. \\

\subsection{Hamiltonian dynamics on Lie algebroids}
\label{subsec-algebroids}

In this section, we review some basic facts about Lie algebroids, see [2] for more details.\\
Let $ E $ be a Lie algebroid over a manifold $ M $, and $ \tau^{\ast} : E^{\ast} \longrightarrow M $ be a vector bundle projection of $ E^{\ast} $ to $ M $, where $ E^{\ast} $ is dual bundle of $ E $. Define the prolongation of $ E $ over $ E^{\ast} $ as follows
$$ \mathcal{T}^{\tau^{\ast}} E = \lbrace (b,v) \in E \times TE^{\ast} \vert \rho(b) = T\tau^{\ast} (v) \rbrace. $$

$ \mathcal{T}^{\tau^{\ast}} E  $ is a Lie algebroid of rank $ 2n $ over $ E^{\ast} $ with the following Lie algebroid structure:\\

If $ ( f_{i} (X_{i} \circ \tau^{\ast}) , X^{'} ) $ and $ ( g_{j} (Y_{j} \circ \tau^{\ast}) , Y^{'} ) $ are two sections of $ \mathcal{T}^{\tau^{\ast}} E $, with \\
 $ f_{i} , g_{j} \in C^{\infty} (E^{\ast}) $, $ X_{i} , Y_{j} \in \Gamma (E) $ and $ X^{'} , Y^{'} \in \mathfrak{X} (E^{\ast}) $, then\\
$$ [ ( f_{i} (X_{i} \circ \tau^{\ast}) , X^{'} ) , ( g_{j} (Y_{j} \circ \tau^{\ast}) , Y^{'} ) ] =( f_{i} g_{j} ([ X_{i} , Y_{j} ] \circ \tau^{\ast})$$
$$ + X^{'}(g_{j}) (Y_{j}\circ \tau^{\ast} ) - Y^{'} (f_{i}) (X_{i} \circ \tau^{\ast} ) , [X^{'} , Y^{'} ] ) $$

\begin{center}
$  \rho  ( f_{i} (X_{i} \circ \tau^{\ast}) , X^{'} ) =X^{'} $
\end{center}
If $ a^{\ast} \in E^{\ast} $ and $ (b , v) \in \mathcal{T}^{\tau^{\ast}} E $ over $ a^{\ast} $, then

\begin{center}
$ \Theta (a^{\ast}) (b,v) = a^{\ast}(b) $
\end{center}
is called \textbf{Liouville section} of $ ( \mathcal{T}^{\tau^{\ast}} E )^{\ast} $, and the \textbf{canonical symplectic section} $ \Omega $ is defined by
\begin{center}
$ \Omega = -d\Theta $
\end{center} 
\begin{proposition}
\label{prop-poissonreduction}
$ \Omega $ is symplectic section of Lie algebroid $ \mathcal{T}^{\tau^{\ast}} E  $, that is $ d \Omega = 0 $ and it is non-degenerate 2-section.
\end{proposition}
\begin{proof}
See [2].
\end{proof}

 Let $ H:E^{\ast} \longrightarrow \mathbb{R} $ be a Hamilton function and $ d H \in \Gamma ( ( \mathcal{T}^{\tau^{\ast}} E )^{\ast} ) $, then there exist a unique section $ \xi_{H} \in \Gamma ( \mathcal{T}^{\tau^{\ast}} E ) ) $ satisfying
 \begin{center}
$ i_{\xi_{H}} \Omega = d H $
\end{center}
Suppose that $( E , [,] , \rho ) $ be a Lie algebroid over manifold $ M $, and $ \gamma $ be a section of $ E^{\ast} $. Consider the morphism $ (\phi_{\gamma} , \gamma ):= ( ( Id, T\gamma \circ \rho ) , \gamma ) $ between the vector bundles $ E $ and $ \mathcal{T}^{\tau^{\ast}} E $ defined by $ ( Id, T\gamma \circ \rho ) , \gamma )(a) = ( a , (T_{x} \gamma ) (\rho(a))$, for $ a \in E_{x} $ and $ x \in M. $

\begin{proposition}
\label{prop-poissonreduction}
If $ \gamma $ be a section of dual bundle $ E^{\ast} $ to $ M $ then $ (\phi_{\gamma} , \gamma ) $ is a morphism between the Lie algebroids $ E $ and $ \mathcal{T}^{\tau^{\ast}} E. $  Moreover\\
(i) $ (\phi_{\gamma} , \gamma )^{\ast} \Theta = \gamma   $\\
(ii) $(\phi_{\gamma} , \gamma )^{\ast} \Omega = -d \Theta. $
\end{proposition}
\begin{proof}
See [2].
\end{proof}

\section{Geometric Hamiltonian-Jacobi theorem of Hamiltonian system on Lie algebroids}
In [1], Abraham and Marsden presented a geometric description of the solution of Hamilton-Jacobi equation for the cotangent bundles as a theorem. After that, in 2004, Manuel De Leon, Juan C.Marrero and Eduardo Martinez [2], proved the Abraham and Marsden's theorem for Hamiltonian systems on Lie algebroids as follows.
\begin{theorem}
\label{thm-variational groupoid}
Let $ E $ be a Lie algebroid over manifold $ M $ and $ H: E^{\ast} \longrightarrow \mathbb{R} $ be a Hamiltoniam function. Consider Lie algebroid structure on $ \mathcal{T}^{\tau^{\ast}} E $ and let $ \xi_{H} \in \Gamma (\mathcal{T}^{\tau^{\ast}} E ) $ be the corresponding Hamiltonian section. Let $ \gamma \in \Gamma(E^{\ast}) $ be a cocycle, $ d \gamma = 0, $ and denote by $ \xi_{H} ^{\gamma} \in \Gamma (E)$ the section $ \xi_{H} ^{\gamma} = pr_{1} \circ \xi \circ\gamma. $ Then the following condition are equivalent:\\
(i) For any curve $ t \longrightarrow \sigma (t) $ in $ M $ satisfying
\begin{center}
$ \rho (\xi_{H} ^{\gamma}) ( \sigma (t) ) = \dot{\sigma} (t) $ ~~~~~~ for all $ t $
\end{center}
the curve $ t \longrightarrow \gamma ( \sigma (t) ) $ on $ E^{\ast} $ satisfies the Hamilton equation for $ H. $ \\
(ii) $ \gamma $ satisfies the Hamilton-Jacobi equation $ d (H\circ \gamma ) = 0. $
\end{theorem}
 \begin{proof}
See [2].
\end{proof}

\begin{definition}
\label{dfn-second order}
The section $ \gamma $ is called to be cocycle with respect to $ pr_{1}  : \mathcal{T}^{\tau^{\ast}} E \longrightarrow E $ if for any $ v , w \in \mathcal{T}^{\tau^{\ast}} E $ we have
 \begin{center}
$  d \gamma ( pr_{1} (v) , pr_{1} (w) ) = 0 . $
\end{center}
\end{definition}

 Now, we prove the key lemma, which in fact is an extension of the results corresponding to Abraham and Marsden in [1] and the previous theorem.

\begin{lemma}
\label{lemma-functorial legendre}
Assume that $ \gamma \in \Gamma (E^{\ast} ) $  is a section of $ E^{\ast} $ and $ \lambda = \gamma \circ \tau^{\ast} : E^{\ast} \longrightarrow E^{\ast} $ and $ \tilde{\mathcal{T}} \lambda := ( \mathcal{T} \lambda \circ \gamma) \circ pr_{1} = (\phi_{\gamma} , \gamma ) \circ pr_{1} : \mathcal{T}^{\tau^{\ast}} E \longrightarrow \mathcal{T}^{\tau^{\ast}} E, $ then we have the following two assertions holds:\\
(i) For any $ x , y \in E $, $ ( \phi_{\gamma} , \gamma)^{\ast} \Omega (x , y) = -d \gamma(x,y),  $ and for every $ v,w \in \mathcal{T}^{\tau^{\ast}} E, $ \\
  $ \tilde{\lambda}^{\ast} \Omega (v,w) = -d \gamma ( pr_{1} (v) , pr_{1} (w) ), $  since $ \Omega $ is the canonical symplectic section of $ \mathcal{T}^{\tau^{\ast}} E. $\\
\\
(ii) For every $ v , w \in \mathcal{T}^{\tau^{\ast}} E,~~~ \Omega ( \tilde{\mathcal{T}} \lambda.v , w ) = \Omega (v , w- \tilde{\mathcal{T}} \lambda.w ) - d\gamma (pr_{1} (v) , pr_{1} (w) ). $

\end{lemma}
\begin{proof}
(i) Since $ \Omega $ is canonical symplectic section of $ \mathcal{T}^{\tau^{\ast}} E, $ we know that there is a unique canonical one-section $ \Theta $ such that $ \Omega = -d \Theta. $ We know that $ (\phi_{\gamma} , \gamma )^{\ast} \Theta = \gamma, $ then we obtain that

\begin{center}
$  (\phi_{\gamma} , \gamma )^{\ast} \Omega (x , y) = (\phi_{\gamma} , \gamma )^{\ast} (-d\Theta ) (x,y) = - d ((\phi_{\gamma} , \gamma )^{\ast} \Theta) (x,y) = -d \gamma (x , y).  $
\end{center}

Note that $ \tilde{\lambda}^{\ast} := pr_{1} ^{\ast} \circ (\phi_{\gamma} , \gamma )^{\ast} : ( \mathcal{T}^{\tau^{\ast}} E )^{\ast} \longrightarrow ( \mathcal{T}^{\tau^{\ast}} E )^{\ast}, $ then we have
\begin{eqnarray*}
\tilde{\lambda}^\ast   \Omega (v,w)&=& \tilde { \lambda}^{\ast} (-d \Theta) (v,w)
\\&=&-d ( \tilde {\lambda}^{\ast} \Theta ) (v , w) \\&=&-d ( pr_{1} ^{\ast} \circ (\phi_{\gamma} , \gamma )^{\ast}  \Theta ) (v , w )\\&=& -d (pr_{1}^{\ast} \circ \gamma ) (v,w)
\\&=&-d \gamma ( pr_{1} (v) , pr_{1} (w) ).
\end{eqnarray*}
(ii) For any $ v , w \in \mathcal{T}^{\tau^{\ast}} E, ~~ v - \tilde{\mathcal{T}} \lambda.v $ is vertical, because

\begin{center}
$ pr_{1} ( \tilde{\mathcal{T}} \lambda.v ) = pr_{1} (v) - ( pr_{1} \circ (\phi_{\gamma}, \gamma ) \circ pr_{1} ) . v
= pr_{1} (v) - pr_{1} (v) = 0$
\end{center}
Where we have used relation $  pr_{1} \circ (\phi_{\gamma}, \gamma ) \circ pr_{1}  = pr_{1}. $ Thus $ \Omega ( v - \tilde{\mathcal{T}} \lambda.v , w- \tilde{\mathcal{T}} \lambda.w ) = 0, $ and hence

 \begin{center}
$ \Omega ( v- \tilde{\mathcal{T}} \lambda.v , w ) = \Omega ( v ,  w- \tilde{\mathcal{T}} \lambda.w  ) + \Omega ( v- \tilde{\mathcal{T}} \lambda.v , w- \tilde{\mathcal{T}} \lambda.w ). $
\end{center}

However the second term on the right hand is given by
\begin{eqnarray}
\label{eq-brack}
\Omega ( v- \tilde{\mathcal{T}} \lambda.v , w- \tilde{\mathcal{T}} \lambda.w ) &=& \tilde{\lambda}^{\ast} \Omega (v,w) =  pr_{1} ^{\ast} \circ (\phi_{\gamma} , \gamma )^{\ast} (v , w )     \nonumber \\
       & = &	  (\phi_{\gamma} , \gamma )^{\ast} \Omega (  pr_{1} (v) , pr_{1} (w)) = d \gamma (  pr_{1} (v) , pr_{1} (w)),   \nonumber
\end{eqnarray}
where we have used the asseration (i). It follows that
\begin{eqnarray}
\label{eq-brack}
 \Omega ( \tilde{\mathcal{T}} \lambda.v , w)  &=& \Omega ( ( (\phi_{\gamma}, \gamma ) \circ pr_{1} ).v , w )     \nonumber \\
       & = &	  \Omega ( v, w - (\phi_{\gamma}, \gamma ) \circ pr_{1} ). w ) - d \gamma ( pr_{1} (v) , pr_{1} (w) )     \nonumber \\
       &  =&  \Omega (v , w- \tilde{\mathcal{T}} \lambda.w ) - d \gamma ( pr_{1} (v) , pr_{1} (w) ).     \nonumber
\end{eqnarray}
So, the assertion (ii) holds.
\end{proof}

Now, for a given Hamiltonian system $ (E^{\ast} , \Omega , H ) $ on $ E $ over $ M $, by using the above Lemma, we can prove the following two types of geometric Hamilton-Jacobi theorem for Hamiltonian system. At first, we prove Type I of geometric Hamilton-Jacobi theorem for Hamiltonian system on Lie algebroids by using the fact that the section $ \gamma \in \Gamma (E^{\ast}) $ is cocycle with respect to $ pr_{1} : \mathcal{T}^{\tau^{\ast}} E \longrightarrow E. $\\
\begin{theorem}
\label{thm-Type I of HamiltonJacobi theorem}
For Hamilton system $ (E^{\ast} , \Omega , H ), $ assume that $ \gamma \in \Gamma (E^{\ast})  $ is a section of $ E^{\ast}, $ and $ \xi_{H} ^{\gamma} = pr_{1} \circ \xi \circ\gamma, $ where $ \xi_{H} $ is the corresponding Hamiltonian section. If section $ \gamma \in \Gamma (E^{\ast}) $ is cocycle with respect to $ pr_{1} : \mathcal{T}^{\tau^{\ast}} E \longrightarrow E, $ then $ \gamma $ is a solution of the equation $ (\phi_{\gamma}, \gamma ) \circ \xi_{H} ^{\gamma} = \xi_{H} \circ \gamma, $ which is called Type I of Hamilton-Jacobi equation for Hamiltonian system $ (E^{\ast} , \Omega , H ). $

$$\xymatrix@!=1cm{E^{\ast}\ar[r]^{\tau^{\ast}} & M \ar[r]^{\gamma}\ar[d]^{\xi_{H} ^{\gamma}} & E^{\ast} \ar[d]^{\xi_{H}} \\
 \mathcal{T}^{\tau^{\ast}} E & E \ar[l]^{(\phi_{\gamma} , \gamma)} &  \mathcal{T}^{\tau^{\ast}} E \ar[l]^{pr_{1}}}$$
\end{theorem}

\begin{proof}
Let $ v = \xi_{H} \circ \gamma \in \mathcal{T}^{\tau^{\ast}} E$ and for any $ w \in \mathcal{T}^{\tau^{\ast}} E, ~~pr_{1}(w) \neq 0, $ from Lemma (ii) we have that
\begin{eqnarray}
\label{eq-brack}
  \Omega ( (\phi_{\gamma}, \gamma ) \circ \xi_{H} ^{\gamma} , w )  &=  &	 \Omega ( (\phi_{\gamma}, \gamma ) \circ \xi_{H} \circ \gamma , w )      \nonumber \\
       & = &	 \Omega ( \xi_{H}  \circ \gamma , w)-  (\phi_{\gamma}, \gamma ) \circ pr_{1} ). w ) - d\gamma ( pr_{1} (\xi_{H} \circ \gamma , pr_{1} (w) )      \nonumber \\
       &  =& \Omega ( \xi_{H}  \circ \gamma , w ) -  \Omega ( \xi_{H}  \circ \gamma , \tilde{\mathcal{T}} \lambda.w ) - d \gamma ( pr_{1}(v) , pr_{1}(w) ).      \nonumber
\end{eqnarray}

Because the section $ \gamma \in \Gamma (E^{\ast}) $ is cocycle with respect to $  pr_{1} : \mathcal{T}^{\tau^{\ast}} E \longrightarrow E, $ then we have that $ d\gamma ( pr_{1} (\xi_{H} \circ \gamma , pr_{1} (w) ) =0, $ so
\begin{eqnarray}
\label{eq-brack}
  \Omega ( (\phi_{\gamma}, \gamma ) \circ \xi_{H} ^{\gamma} , w )  -	 \Omega ( \xi_{H} \circ \gamma , w ) = \Omega ( \xi_{H}  \circ \gamma , \tilde{\mathcal{T}} \lambda.w ).
\end{eqnarray}
If $ \gamma $ satisfies the equation $ (\phi_{\gamma}, \gamma ) \circ \xi_{H} ^{\gamma} = \xi_{H} \circ \gamma, $ from Lemma (i) we can obtain that
\begin{eqnarray}
\label{eq-brack}
  -\Omega (  \xi_{H}  \circ \gamma , \tilde{\mathcal{T}} \lambda.w  ) &=  &	 \Omega ( (\phi_{\gamma}, \gamma ) \circ \xi_{H} ^{\gamma} , \tilde{\mathcal{T}} \lambda.w )      \nonumber \\
   & = &	 -\Omega (  (\phi_{\gamma}, \gamma ) \circ pr_{1} \circ \xi_{H} \circ \gamma , \tilde{\mathcal{T}} \lambda.w )      \nonumber \\
       &  =& -\Omega ( \tilde{\mathcal{T}} \lambda \circ \xi_{H}  \circ \gamma , \tilde{\mathcal{T}} \lambda.w )       \nonumber  \\
       &=&  \tilde{\lambda}^{\ast} \Omega ( \xi_{H}  \circ \gamma , w )   \nonumber   \\
       &=&   d\gamma ( pr_{1} (\xi_{H} \circ \gamma ) , pr_{1} (w) )  \nonumber \\
       &=& 0 \nonumber
\end{eqnarray}
But, since the symplectic section $ \Omega $ is non-degenerate, the left side of (1) equals zero, only when $ \Omega $ satisfies $ (\phi_{\gamma}, \gamma ) \circ \xi_{H} ^{\gamma} = \xi_{H} \circ \gamma, $ thus the section $ \gamma \in \Gamma (E^{\ast}) $ is cocycle with respect to $ pr_{1} : \mathcal{T}^{\tau^{\ast}} E \longrightarrow E, $ then $ \gamma $ must be a solution of Type I of Hamilton-Jacobi equation $ (\phi_{\gamma}, \gamma ) \circ \xi_{H} ^{\gamma} = \xi_{H} \circ \gamma. $
\end{proof}

In the following, we show that for every symplectic morphism $ \varepsilon : E^{\ast} \longrightarrow E^{\ast} $, we can prove the following Type II of geometric Hamilton-Jacobi theorem for the Hamiltonian system on Lie algebroids.

\begin{theorem}
\label{thm-Type II of HamiltonJacobi theorem}
For Hamilton system $ (E^{\ast} , \Omega , H ), $ assume that $ \gamma \in \Gamma (E^{\ast})  $ is a section of $ E^{\ast}, $ and $ \tilde{\mathcal{T}} \lambda := ( \mathcal{T} \lambda \circ \gamma) \circ pr_{1} = (\phi_{\gamma} , \gamma ) \circ pr_{1} : \mathcal{T}^{\tau^{\ast}} E \longrightarrow \mathcal{T}^{\tau^{\ast}} E , $ and for every symplectic morphism $ \varepsilon : E^{\ast} \longrightarrow E^{\ast}, $ denote by $ \xi_{H} ^{\gamma} = pr_{1} \circ \xi \circ\gamma, $ where $ \xi_{H} $ is the corresponding Hamiltonian section. Then $ \varepsilon $ is a solution of the equation $ \mathcal{T} \varepsilon \circ \xi_{H} = \tilde{\mathcal{T}} \lambda \circ\xi_{H} \circ \varepsilon, $ if and only if it is a solution of the equation $ (\phi_{\gamma}, \gamma ) \circ \xi_{H} ^{\varepsilon} = \xi_{H} \circ \varepsilon, $ where $ \xi_{H \circ\varepsilon} \in \Gamma ( \mathcal{T}^{\tau^{\ast}} E ) $ is Hamiltonian section of the function $ H \circ \varepsilon : E^{\ast} \longrightarrow \mathbb{R}. $The equation $ (\phi_{\gamma}, \gamma ) \circ \xi_{H} ^{\varepsilon} = \xi_{H} \circ \varepsilon, $ is called Type II of Hamilton-Jacobi equation for Hamiltonian system $ (E^{\ast} , \Omega , H ). $

$$\xymatrix@!=1cm{ E^{\ast} \ar[r] & E^{\ast}\ar[r]^{\tau^{\ast}} \ar[d]^{\xi_{H\circ \varepsilon}} \ar[dr]^{\xi_{H} ^{\varepsilon}} & M \ar[r]^{\gamma} & E^{\ast} \ar[d]^{\xi_{H}}\\
 \mathcal{T}^{\tau^{\ast}} E & \mathcal{T}^{\tau^{\ast}} E \ar[l]^{\mathcal{T} \varepsilon} & E \ar[l]^{(\phi_{\gamma} , \gamma)} &  \mathcal{T}^{\tau^{\ast}} E \ar[l]^{pr_{1}}}$$

\end{theorem}
 \begin{proof}

Let $ v = \xi_{H} \circ \varepsilon \in \mathcal{T}^{\tau^{\ast}} E$ and for any $ w \in \mathcal{T}^{\tau^{\ast}} E, ~~\tilde{\mathcal{T}} \lambda(w) \neq 0, $ from Lemma  we have that
\begin{eqnarray}
\label{eq-brack}
  \Omega ( (\phi_{\gamma}, \gamma ) \circ \xi_{H} ^{\varepsilon} , w )  &=  &	 \Omega ( (\phi_{\gamma}, \gamma ) \circ pr_{1} \circ \xi_{H} \circ \varepsilon , w )      \nonumber \\
       & = &	 \Omega ( \xi_{H}  \circ \varepsilon , w - ( (\phi_{\gamma}, \gamma ) \circ pr_{1} ). w ) - d\gamma ( pr_{1} (\xi_{H} \circ \varepsilon ) , pr_{1} (w) )      \nonumber \\
       &  =& \Omega ( \xi_{H}  \circ \varepsilon , w ) -  \Omega ( \xi_{H}  \circ \varepsilon , \tilde{\mathcal{T}} \lambda.w ) +  \tilde{\lambda}^{\ast} \Omega ( \xi_{H}  \circ \varepsilon , w )      \nonumber   \\
         &=&  \Omega ( \xi_{H}  \circ \varepsilon, w ) -  \Omega ( \xi_{H}  \circ \varepsilon,  \tilde{\mathcal{T}} \lambda.w) + \Omega ( \tilde{\mathcal{T}} \lambda \circ \xi_{H}  \circ \varepsilon , \tilde{\mathcal{T}} \lambda\circ \varepsilon ).    \nonumber
\end{eqnarray}

Because that $ \varepsilon : E^{\ast} \longrightarrow E^{\ast} $ is symplectic, ana so $ \xi_{H} \circ \varepsilon = \mathcal{T} \varepsilon \circ \xi_{H}. $ From the reasoning above, we conclude that
\begin{eqnarray*}
\Omega ( ( \phi_{\gamma}, \gamma ) \circ \xi_{H} ^{\gamma} , w )-\Omega ( \xi_{H} \circ \gamma , w )
&=& - \Omega ( \mathcal{T} \varepsilon \circ \xi_{H} , \tilde{\mathcal{T}} \lambda.w ) + \Omega ( \tilde{\mathcal{T}} \lambda \circ \xi_{H} \circ \varepsilon , \tilde{\mathcal{T}} \lambda.w )\\
&=& \Omega ( \tilde{\mathcal{T}} \lambda \circ \xi_{H}  \circ \gamma - \mathcal{T} \varepsilon \circ \xi_{H}  , \tilde{\mathcal{T}} \lambda.w ).
\end{eqnarray*}
Since the symplectic section $ \Omega $ is non-degenerate, it follows that $ (\phi_{\gamma}, \gamma ) \circ \xi_{H} ^{\varepsilon} = \xi_{H} \circ \varepsilon $ is equivalent to $ \mathcal{T} \varepsilon \circ \xi_{H} = \tilde{\mathcal{T}} \lambda \circ\xi_{H} \circ \varepsilon, $ if and only if it is a solution of Type II of Hamilton-Jacobi equation $ (\phi_{\gamma}, \gamma ) \circ \xi_{H} ^{\varepsilon} = \xi_{H} \circ \varepsilon. $

\end{proof}

\subsection{Examples}
\label{subsec-secondexamples}
(i) If $ E $ be the tangent lie algebroid $ TQ,$ then the Lie algebroid  $ (\mathcal{T}^{\tau^{\ast}} E, [,] , \rho) $ is the standard Lie algebroid $  (T( T^{\ast}Q) , [,], Id).$ If assume that $ \gamma: Q \longrightarrow T^{\ast}Q $ be a one-form on $ Q, $ and $ X_{H} ^{\gamma} = T \pi_{Q} \circ X_{H} \circ \gamma $ where $ X_{H} $ is dynamical vector field of Hamiltonian system $ (T^{\ast} Q , \omega , H ).$ If $ \gamma $ be closed with respect to $ T \pi_{Q} : T (T^{\ast} Q) \longrightarrow TQ,$ then $ \gamma $ is the solution of Type I of Hamilton-Jacobi equation  $ T\gamma \circ X_{H} ^{\gamma} = X_{H} \circ \gamma.$ \\

$$\xymatrix@!=2cm{T^{\ast}Q \ar[r]^{\pi_{Q}} & Q \ar[r]^{\gamma}\ar[d]^{X_{H} ^{\gamma}} & T^{\ast}Q \ar[d]^{X_{H}}\\
T (T^{\ast} Q) & TQ \ar[l]^{T\gamma} &  T (T^{\ast} Q) \ar[l]^{T\pi_{Q} }}$$\\

For any symplectic map $ \varepsilon : T^{\ast} Q \longrightarrow T^{\ast} Q, ~~\varepsilon $ is a solution of the equation $ T \varepsilon. X_{H\circ\varepsilon} = T \lambda \circ X_{H} \circ \varepsilon $ if and only if it is a solution of Type II of Hamilton-Jacobi equation $ T\gamma \circ X_{H} ^{\varepsilon} = X_{H} \circ \varepsilon. $

$$\xymatrix@!=2cm{ T^{\ast}Q \ar[r] & T^{\ast}Q \ar[r]^{\pi_{Q}} \ar[d]^{X_{H\circ \varepsilon}} \ar[dr]^{X_{H} ^{\varepsilon}} & Q \ar[r]^{\gamma} & T^{\ast}Q \ar[d]^{X_{H}}\\
 T (T^{\ast} Q) & T (T^{\ast} Q) \ar[l]^{T \varepsilon} & T Q \ar[l]^{T\gamma} &  T (T^{\ast} Q) \ar[l]^{T\pi_{Q}}}$$\\

(ii) Let $ ( E= TQ/G , [,], \rho), ~~ \tau_{q\vert G} : TQ/G \longrightarrow Q/G, $ be gauge algebroid associated with the principle bundle $ \pi: Q \longrightarrow M=Q/G. $ Since $ \mathcal{T}^{(\tau_{Q\vert G} )^{\ast}} (TQ/G) \cong T (T^{\ast} Q/G) $ and $  (\mathcal{T}^{(\tau_{Q\vert G} )^{\ast}} (TQ/G))^{\ast} \cong T^{\ast} (T^{\ast} Q/G) $,then if assume that $ \gamma_{Q \vert  G}: Q/G \longrightarrow T^{\ast}Q/G $ be a one-form on $ Q/G, $ and $ X_{H_{Q \vert G}} ^{\gamma_{Q \vert G}} = T \pi_{Q \vert G} \circ X_{H_{Q \vert G}} \circ \gamma_{Q \vert G} $ where $ X_{H_{Q \vert G}} $ is dynamical vector field of Hamiltonian system $ (T^{\ast} Q/G , \omega , H_{Q \vert G} ).$ If $ \gamma_{Q \vert G} $ be closed with respect to $ T \pi_{Q} : T T^{\ast} Q/G \longrightarrow TQ/G,$ then $ \gamma_{Q \vert G} $ is the solution of Type I of Hamilton-Jacobi equation  $ T\gamma_{Q \vert G} \circ X_{H_{Q \vert G}} ^{\gamma_{Q \vert G}} = X_{H_{Q \vert G}} \circ \gamma_{Q \vert G}.$ \\

$$\xymatrix@!=2cm{T^{\ast}Q/G \ar[r]^{\pi_{Q \vert G}} & Q/G \ar[r]^{\gamma_{Q \vert G}}\ar[d]^{X_{H_{Q \vert G}} ^{\gamma_{Q \vert G}}} & T^{\ast}Q/G \ar[d]^{X_{H_{Q \vert G}}}\\
T (T^{\ast} Q/G) & TQ/G \ar[l]^{T\gamma_{Q \vert G}} &  T (T^{\ast} Q/G) \ar[l]^{T\pi_{Q \vert G} }}$$\\

 For any symplectic map $ \varepsilon : T^{\ast} Q/G \longrightarrow T^{\ast} Q/G, ~~\varepsilon $ is a solution of the equation $ T \varepsilon. X_{H_{Q \vert G} \circ \varepsilon} = T \lambda \circ X_{H_{Q \vert G}} \circ \varepsilon $ if and only if it is a solution of Type II of Hamilton-Jacobi equation $ T\gamma \circ X_{H_{Q \vert G}} ^{\varepsilon} = X_{H_{Q \vert G}} \circ \varepsilon. $\\

$$\xymatrix@!=2cm{ T^{\ast}Q/G \ar[r] & T^{\ast}Q/G \ar[r]^{\pi_{Q \vert G}} \ar[d]^{X_{H_{Q \vert G} \circ \varepsilon}} \ar[dr]^{X_{H_{Q \vert G}} ^{\varepsilon}} & Q/G \ar[r]^{\gamma_{Q \vert G}} & T^{\ast}Q/G \ar[d]^{X_{H_{Q \vert G}}}\\
 T (T^{\ast} Q/G) & T (T^{\ast} Q/G) \ar[l]^{T \varepsilon} & T Q/G \ar[l]^{T\gamma_{Q \vert G}} &  T (T^{\ast} Q/G) \ar[l]^{T\pi_{Q \vert G}}}$$

\section{Time-dependent Hamilton-Jacobi theorem of Hamiltonian system on Lie algebroid}

\label{subsec-Time dependent}
In the previous section we have introduced Type I and II of Hamilton-Jacobi theory for Hamiltonnian system $ (E^{\ast} , H, \Omega ). $ In this section we want to obtain these two Types of theory when they are dependent on time.
\\To achieve this goal, we need to introduce the extended-formalism. This formalism, provides away to deal with the space $ \mathbb{R} \times E^{\ast}, $ introduced to deal with the time-dependent Hamiltonians.  The solution is carry the dynamics to $  (\mathbb{R} \times E)^{\ast}. $

On $ \mathbb{R} \times E $ we consider coordinates given by the product of the global coordinates on $ \mathbb{R}, ~t, $ and coordinate $ (x , \lambda ) $ on $ E. $  Then $ (\mathbb{R} \times E)^{\ast}$ can be endowed with natural coordinates $ t, e, x, \mu $, where $ e $ is the t-conjugated momentum. There is a natural map $ \pi: ( \mathbb{R} \times E)^{\ast} \longrightarrow \mathbb{R} \times E^{\ast} $ given in local coordinates by $ \pi (t , e , x , \mu) = ( t , x , \mu). $\\

Suppose that $( \mathbb{R} \times E , [,] , \rho_{\mathbb{R}_{t}} ) $ be a Lie algebroid over manifold $ \mathbb{R} \times M $, where $ \rho_{ \mathbb{R}_{t}}(a) = \rho_{ \mathbb{R}} (t,a) := ( t, \rho (a) ) $ and $ \gamma_{\mathbb{R}_{t}}(x) = \gamma_{\mathbb{R}} (t , x) :=( t , \gamma (x)), $ where $ \gamma \in \Gamma (E^{\ast}) $ and $ \gamma_{\mathbb{R}} \in \Gamma (\mathbb{R} \times E^{\ast} ), $ be a section of $  E^{\ast} $. Consider the morphism $ (\phi_{\gamma_{\mathbb{R}_{t}}} , \gamma_{\mathbb{R}_{t}} ):= ( ( Id , T\gamma_{\mathbb{R}_{t}} \circ \rho_{\mathbb{R}_{t}} ) , \gamma_{\mathbb{R}_{t}} ) $ between the vector bundles $  E $ and $ \mathcal{T}^{\tau ^{\ast}}  E  $ defined by $ ( Id, T\gamma_{\mathbb{R}_{t}} \circ \rho_{\mathbb{R}_{t}} ) , \gamma_{\mathbb{R}_{t}} )( a) = ( a , (T_{x} \gamma_{\mathbb{R}_{t}} ) (\rho_{\mathbb{R}_{t}}(a))$, for $ ( a) \in E_{x} $ and $ x \in  M. $\\

At first, by using the symbols described above, we express theorem 5 as time-dependent as follows.\\
\begin{theorem}
\label{thm-variational groupoid}
Let $ E $ be a Lie algebroid over manifold $ M $ and $ H: E^{\ast} \longrightarrow \mathbb{R} $ be a Hamiltoniam function. Consider Lie algebroid structure on $ \mathcal{T}^{\tau^{\ast}} E $ and let $ \xi_{H} \in \Gamma (\mathcal{T}^{\tau^{\ast}} E ) $ be the corresponding Hamiltonian section. Let $ \gamma_{\mathbb{R}_{t}} \in \Gamma(E^{\ast}) $ be a cocycle, $ d \gamma_{\mathbb{R}_{t}} = 0, $ and denote by $ \xi_{H} ^{\gamma_{\mathbb{R}_{t}}} \in \Gamma (E)$ the section $ \xi_{H} ^{\gamma} = pr_{1} \circ \xi \circ\gamma_{\mathbb{R}_{t}}. $ Then the following condition are equivalent:\\
(i) For any curve $ t \longrightarrow \sigma (t) $ in $ M $ satisfying
\begin{center}
$ \rho (a) ( \sigma (t) ) = \dot{\sigma} (t) $ ~~~~~~ for all $ t $
\end{center}
the curve $ t \longrightarrow \gamma_{\mathbb{R}_{t}} ( \sigma (t) ) $ on $ E^{\ast} $ satisfies the Hamilton equation for $ H. $ \\
(ii) $ \gamma_{\mathbb{R}_{t}} $ satisfies the Hamilton-Jacobi equation.
\end{theorem}

Now, Type I of Hamilton-Jacobi theory for non-autonomous case have the following forms:\\

\begin{theorem}
\label{thm-Type I of HamiltonJacobi theorem}
For Hamilton system $ (E^{\ast} , \Omega , H ), $ assume that $ \gamma_{\mathbb{R}_{t}} \in \Gamma (E^{\ast})  $ is a section of $ E^{\ast}, $ and $ \xi_{H} ^{\gamma_{\mathbb{R}_{t}}} = pr_{1} \circ \xi \circ \gamma_{t}, $ where $ \xi_{H} $ is the corresponding Hamiltonian section. If section $ \gamma_{\mathbb{R}_{t}} \in \Gamma (E^{\ast}) $ is cocycle with respect to $ pr_{1} : \mathcal{T}^{\tau^{\ast}} E \longrightarrow E, $ then $ \gamma_{\mathbb{R}_{t}} $ is a solution of the equation $ (\phi_{\gamma_{\mathbb{R}_{t}}}, \gamma_{t} ) \circ \xi_{H} ^{\gamma_{\mathbb{R}_{t}}} = \xi_{H} \circ \gamma_{\mathbb{R}_{t}}, $ which is called Type I of Hamilton-Jacobi equation for Hamiltonian system $ (E^{\ast} , \Omega , H ). $

$$\xymatrix@!=2cm{E^{\ast}\ar[r]^{\tau^{\ast}} & M \ar[r]^{\gamma_{\mathbb{R}_{t}}}\ar[d]^{\xi_{H} ^{\gamma_{\mathbb{R}_{t}}}} & E^{\ast} \ar[d]^{\xi_{H}}\\
 \mathcal{T}^{\tau^{\ast}} E & E \ar[l]^{(\phi_{\gamma_{\mathbb{R}_{t}}} , \gamma_{\mathbb{R}_{t}})} &  \mathcal{T}^{\tau^{\ast}} E \ar[l]^{pr_{1}}}$$
\end{theorem}

In the following, we show that for every symplectic morphism $ \varepsilon : E^{\ast} \longrightarrow E^{\ast} $, we can prove the following Type II of geometric Hamilton-Jacobi theorem for the Hamiltonian system.

\begin{theorem}
\label{thm-Type II of HamiltonJacobi theorem}
For Hamilton system $ (E^{\ast} , \Omega , H ), $ assume that $ \gamma_{\mathbb{R}_{t}} \in \Gamma (E^{\ast})  $ is a section of $ E^{\ast}, $ and $ \tilde{\mathcal{T}} \lambda := ( \mathcal{T} \lambda \circ \gamma_{\mathbb{R}_{t}}) \circ pr_{1} = (\phi_{\gamma_{\mathbb{R}_{t}}} , \gamma_{\mathbb{R}_{t}} ) \circ pr_{1} : \mathcal{T}^{\tau^{\ast}} E \longrightarrow \mathcal{T}^{\tau^{\ast}} E , $ and for every symplectic morphism $ \varepsilon : E^{\ast} \longrightarrow E^{\ast}, $ denote by $ \xi_{H} ^{\gamma_{\mathbb{R}_{t}}} = pr_{1} \circ \xi \circ\gamma_{\mathbb{R}_{t}}, $ where $ \xi_{H} $ is the corresponding Hamiltonian section. Then $ \varepsilon $ is a solution of the equation $ \mathcal{T} \varepsilon \circ \xi_{H} = \tilde{\mathcal{T}} \lambda \circ\xi_{H} \circ \varepsilon, $ if and only if it is a solution of the equation $ (\phi_{\gamma_{\mathbb{R}_{t}}}, \gamma_{\mathbb{R}_{t}} ) \circ \xi_{H} ^{\varepsilon} = \xi_{H} \circ \varepsilon, $ where $ \xi_{H \circ\varepsilon} \in \Gamma ( \mathcal{T}^{\tau^{\ast}} E ) $ is Hamiltonian section of the function $ H \circ \varepsilon : E^{\ast} \longrightarrow \mathbb{R}. $The equation $ (\phi_{\gamma_{t}}, \gamma_{\mathbb{R}_{t}} ) \circ \xi_{H} ^{\varepsilon} = \xi_{H} \circ \varepsilon, $ is called Type II of Hamilton-Jacobi equation for Hamiltonian system $ (E^{\ast} , \Omega , H ). $

$$\xymatrix@!=2cm{ E^{\ast} \ar[r] & E^{\ast}\ar[r]^{\tau^{\ast}} \ar[d]^{\xi_{H\circ \varepsilon}} \ar[dr]^{\xi_{H} ^{\varepsilon}} & M \ar[r]^{\gamma_{\mathbb{R}_{t}}} & E^{\ast} \ar[d]^{\xi_{H}}\\
 \mathcal{T}^{\tau^{\ast}} E & \mathcal{T}^{\tau^{\ast}} E \ar[l]^{\mathcal{T} \varepsilon} & E \ar[l]^{(\phi_{\gamma_{\mathbb{R}_{t}}} , \gamma_{t})} &  \mathcal{T}^{\tau^{\ast}} E \ar[l]^{pr_{1}}}$$

\end{theorem}

\small{

}
\end{document}